\documentclass[fontsize=11pt]{article}

\usepackage[english]{babel}

\usepackage[letterpaper,top=1in,bottom=1in,left=1in,right=1in]{geometry}

\usepackage{amsmath}
\usepackage{graphicx}
\usepackage[colorlinks=true, allcolors=blue]{hyperref}
\usepackage{natbib}
\usepackage{amssymb, amsthm}
\usepackage{booktabs}
\usepackage{multirow}
\usepackage{authblk}
\usepackage{lineno}
\usepackage{url}

\usepackage{listings}
\usepackage{xcolor}

\definecolor{codegreen}{rgb}{0,0.6,0}
\definecolor{codegray}{rgb}{0.5,0.5,0.5}
\definecolor{codepurple}{rgb}{0.58,0,0.82}
\definecolor{backcolour}{rgb}{0.95,0.95,0.92}

\lstdefinestyle{mystyle}{
    commentstyle=\color{codegreen},
    keywordstyle=\bf,
    stringstyle=\color{magenta},
    basicstyle=\ttfamily,
    breakatwhitespace=false,         
    breaklines=true,                 
    captionpos=b,                    
    keepspaces=true,                 
    numbers=none,                    
    numbersep=5pt,                  
    showspaces=false,                
    showstringspaces=false,
    showtabs=false,                  
    tabsize=2
}

\usepackage{subcaption}
\DeclareCaptionFormat{custom}
{%
    \textbf{#1#2}\textit{#3}
}
\usepackage{xcolor}
\newcommand{\ob}[1]{{#1}}

\newtheorem{lemma}{Lemma}

\theoremstyle{definition}

\newtheorem{example}{Example}
\newtheorem{remark}{Remark}

\title{Estimating the Performance of Entity Resolution Algorithms: Lessons Learned Through PatentsView.org}

\author[1,2]{Olivier Binette}
\author[2]{Sokhna A York}
\author[2]{Emma Hickerson}
\author[1]{Youngsoo Baek}
\author[2]{Sarvo Madhavan}
\author[2]{Christina Jones}
\affil[1]{Duke University}
\affil[2]{American Institutes for Research}

\date{\today}

\begin{document}
\maketitle

\begin{abstract}
This paper introduces a novel evaluation methodology for entity resolution algorithms. It is motivated by PatentsView.org, a public-use patent data exploration platform that disambiguates patent inventors using an entity resolution algorithm. We provide a data collection methodology and tailored performance estimators that account for sampling biases. Our approach is simple, practical and principled -- key characteristics that allow us to paint the first representative picture of PatentsView's disambiguation performance. The results are used to inform PatentsView's users of the reliability of the data and to allow the comparison of competing disambiguation algorithms.
\end{abstract}

\section{Introduction}
Entity resolution (also called record linkage, deduplication, or disambiguation) is the task of identifying records in a database that refer to the same entity. An entity may be a person, a company, an object or an event. Records are assumed to contain partially identifying information about these entities. When there is no unique identifier (such as a social security number) available for all records, entity resolution becomes a complex problem which requires sophisticated algorithmic solutions \citep{Herzog2007, Christen2012, Dong2015, Ilyas2019, Christophides2019, Christen2019, Papadakis2021, Binette2022a}.

{Specifically, we are interested in the entity resolution system used by PatentsView.org, a public-use patent data exploration platform maintained by the American Institutes for Research (AIR), to disambiguate inventor mentions in patent data. The U.S. Patents and Trademarks Office (USPTO) makes available patent data dating back to 1790 (digitized full-text data is available from 1976). However, there is no standard for uniquely identifying inventors on patent applications. The result is a set of ambiguous mentions of inventors, where a single person's name may be spelled in different ways on two applications and where two different inventors with the same name may be difficult to distinguish. Inventors moving between locations and employers further complicates their identification. Following seminal works \citep{trajtenberg2008identification, ferreira2012brief, ventura2013methods, Li2014}, a disambiguation competition was held in 2015 leading to the disambiguation system currently used for PatentsView.org. Since then, disambiguated inventor data has been one of PatentsView's most popular data products, complementing its  data visualizations, search tools, and Application Programming Interface (API), and other data products that serve a wide variety of audience including students, educators, researchers, policymakers, small business owners, and the public \citep{toole2021patentsview}. Given challenges associated with the disambiguation process, the topic continues to be an active area of research \citep{Ventura2015, kim2016random, yang2017mixture, morrison2017disambiguation, muller2017semantic, traylor2017learning, balsmeier2018machine,tam2019optimal, monath2019scalable, doherr2021disambiguation}.}

{One key challenge in using, maintaining, and improving entity resolution systems is to evaluate their performance. In the case of PatentsView's disambiguation, no principled evaluation methodology is available to measure performance, to inform users of the reliability of the data, and to support methodological research to improve upon PatentsView's disambiguation algorithms. The state-of-the-art in entity resolution evaluation, namely computing performance evaluation metrics (precision, recall, etc) on benchmark datasets, leads to misleading and highly biased performance metrics as shown in section \ref{sec:challenges}. This is concerning given the many scientific uses of PatentsView's data: prior to June 2021, 179 research studies cited PatentsView as a data source, including around 25\% from the field of economics \citep{toole2021patentsview}. Furthermore, a common theme of research is the study of the relationship between public policy, inventor mobility, and inventor demographics, on innovation and patenting. This requires accurate inventor disambiguation to track inventors and entities through the breadth of patent data.}

{Our paper addresses this challenge, thus informing users of the reliability of disambiguated data and supporting methodological research to improve disambiguation algorithms. We propose novel evaluation methodology that is principled and cost-effective, and we demonstrate its effectiveness to evaluate PatentsView's disambiguation.}

{Before continuing with the rest of this introduction, we review terminology used throughout in section \ref{sec:terminology}. In section \ref{sec:challenges}, we then continue with challenges of evaluation, past work, and our contributions.}

\subsubsection{Terminology}\label{sec:terminology}

{
We consider a database of \textit{records}, where each record represents a mention to a given inventor (e.g., the first inventor of Patent number 12345). In this context, the records are also referred to as \textit{inventor mentions}. The goal of entity resolution is to \textit{cluster} inventor mentions according to the \textit{entity} (real-world inventor) to which they refer. Clusterings obtained from algorithms are referred to as \textit{predicted} clusters or \textit{predicted} disambiguations, whereas the (unknown) clustering corresponding to the true set of inventors is referred to as the \textit{ground truth}. Two inventor mentions are said to \textit{match} or to be a \textit{true match} if they refer to the same inventor. If two inventors are in the same predicted cluster, then they are a \textit{link}, or a \textit{predicted match}. The proportion of true matches among all predicted matches is called the \textit{pairwise precision}, while the proportion of predicted matches among all true matches is called the \textit{pairwise recall}.
}

\subsection{The Evaluation Problem}\label{sec:challenges}

The entity resolution evaluation problem is to extrapolate from observed performance in small samples to real performance in a database with millions of records. \cite{Wang2022} refer to this as bridging the reality-ideality gap in entity resolution, where high performance on benchmark datasets often does not translate into the real world. Here, performance may be defined as any combination of commonly used evaluation metrics for entity resolution, such as precision and recall, cluster homogeneity and completeness, rand index, or generalized merge distance \citep{Maidasani2012}. These metrics can be computed on benchmark datasets for which we have a ground truth disambiguation. However, the key evaluation problem is to obtain estimates that are representative of performance on the full data, for which no ground truth disambiguation is available. This is challenging for the following reasons.

First, entity resolution problems {do not scale linearly}. While it may be easy to disambiguate a small dataset, the opportunity for errors grows \textit{quadratically} in the number of records. As such, we may observe good performance of an algorithm on a small benchmark dataset, while the true performance on the entire dataset may be something else entirely. This particular effect of dataset size in entity resolution is explored in \cite{Draisbach2013} in the context of choosing similarity thresholds. This is a problem that PatentsView.org currently faces. Despite encouraging performance evaluation metrics on benchmark datasets, with nearly perfect precision and recall reported in the latest methodological report \citep{Monath2021}, the data science team at AIR observes lower real-world accuracy. {This phenomenon is illustrated in example \ref{first_example} below.}

A second problem is large class imbalance in entity resolution \citep{Marchant2017}. {Viewing entity resolution as a classification problem, the task is to classify record pairs as being a match or non-match. However,} among all pairs of records, only a small fraction (usually much less than a fraction of a percent) refer to the same entity. The vast majority of record pairs are not a match. This makes it difficult to evaluate performance through random sampling of record pairs.

A third problem is the multiplicity of sampling mechanisms used to obtain benchmark datasets. To construct hand-disambiguated datasets, blocks, entity clusters, or predicted clusters may be sampled with various probability weights. These sampling approaches must be accounted for in order to obtain representative performance estimates \citep{fuller2011sampling}.

Our approach, detailed in sections \ref{sec:approach} and \ref{sec:proposed_estimators}, addresses these challenges by putting forward novel cluster-based expressions for performance metrics that reflect various sampling schemes. Each of these representations immediately suggests simple estimators that properly account for the above issues.

\begin{example}[Bias of precision computed on benchmark datasets]\label{first_example}
    To exemplify the problem with the trivial use of performance evaluation metrics on benchmark datasets, we carried out a {toy} experiment that is described in detail in appendix \ref{appendix:example}. In short, we evaluated a disambiguation algorithm by sampling ground truth clusters and computing pairwise precision on this set of sampled clusters. This is analogous to the way that many real-world benchmark datasets are obtained and typically used. In this experiment, we know that the disambiguation algorithm has a precision of $52\%$ for the entire dataset.
    
    \begin{figure}[h]
        \centering
        \includegraphics[width=5.5in]{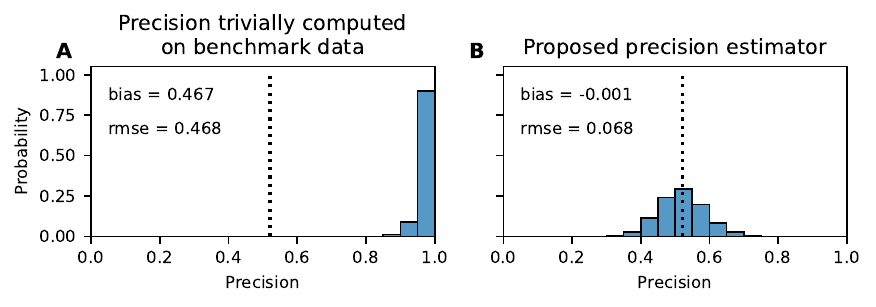}
        \caption{Distribution of precision estimates versus the true precision of $52\%$ (shown as a dotted vertical line). Panel \textbf{A} shows the trivial precision estimates computed for sampled records. Panel \textbf{B} shows our proposed precision estimates which accounts for the sampling mechanism. Sample bias and root mean squared error (rmse) are reported in each figure.}
        \label{fig:precision_problem}
    \end{figure}
    
    In panel \textbf{A} of figure \ref{fig:precision_problem}, we see the distribution of precision estimates versus the true precision of $52\%$ shown as a dotted vertical line. Precision estimates are usually very close to $100\%$ and always higher than $80\%$, despite the truth being a precision of only $52\%$. In contrast, panel \textbf{B} shows the distribution of our proposed precision estimator which is nearly unbiased. Both precision estimators rely on exactly the same data. They only differ in how they account for the underlying sampling process and the extrapolation from small benchmark datasets to the full data.

    {The same phenomenon can be observed in PatentsView's data, where naive precision is nearly 1 on all benchmark datasets. In our simulation studies (see Figure \ref{fig:sim_records_precision} for instance), naive precision estimates are always nearly 1, despite the true precision ranging from 60\% to 90\%.}

    {The reason why naive performance estimation performs disastrously is that it is much easier to disambiguate a small benchmark dataset than a large population with millions of records. Indeed, as a dataset grows, opportunity for erroneous links grows quadratically. False links between similarly named inventor, which are common in the full data, disappear when the benchmark dataset only contains a random sample of inventors. Our performance estimators extrapolate from performance observed on a small benchmark to true performance on the full data.}
    
\end{example}

\subsubsection{Why Bother With Evaluation?}

There are two main uses for accurate and statistically rigorous evaluation methodology.

The first is model selection and comparison. PatentsView.org continually works at improving disambiguation methodology. This requires choosing between alternative methods and evaluating the results of methodological experiments. Without sound evaluation methodology, decisions regarding the disambiguation algorithm may not align with real-world use and real-world performance. Notably, for a performance metric such as the pairwise f-score, one algorithm may perform better than another on a small benchmark dataset, while the opposite may hold true for performance on the entire data. This problem arises with typical benchmark datasets obtained from randomly sampling blocks or randomly sampling clusters (see section \ref{sec:proposed_estimators} for a definition of different sampling mechanisms).

The second is adequate use of disambiguated data. PatentsView.org's disambiguation results have been used in numerous scientific studies \citep{toole2021patentsview}.
{
For example, \cite{Choudhury2019} studied the effect of skilled worker immigration on patenting at U.S. companies and institutions, using PatentsView's inventor disambiguation to track individual immigrant inventors across time and location.
}
These studies make assumptions about the reliability of the data that need to be validated and upheld. In short, users of disambiguated data need to understand its reliability in order to make scientifically appropriate use of it. Evaluation aims to provide this rigorous reliability information.

\subsubsection{Past Work}

Much of the past literature has focused on defining and using relevant clustering evaluation metrics. The topic of estimating performance from samples has received much less attention, usually focusing on importance sampling estimators based on record pairs. We review the contributions to these two main topics below.

\paragraph{Metrics}

Pairwise precision and recall metrics were first reported in \cite{Newcombe1959}, with \cite{Bilenko2003} and \cite{Christen2007} emphasizing the importance of precision-recall curves for algorithm evaluation. However, there are issues with the use of pairwise precision and recall in entity resolution applications, such as the large relative importance of large clusters. As such, other clustering metrics have been proposed, including cluster precision and recall, cluster homogeneity and completeness, the B$^3$ metric \citep{bagga1998algorithms}, and generalized merge distances \citep{Michelson2009, Menestrina2010, Maidasani2012, Barnes2015}. Important practical issues regarding the use of aggregate metrics are discussed in \cite{Hand2018}.

{While our work focuses on estimating pairwise precision and recall, the general approach also applies to any cluster-based performance evaluation metrics, such as those described in \cite{Michelson2009}.}

\paragraph{Estimation}

Regarding the reliable estimation of performance metrics, \cite{Belin1995} first proposed a semi-supervised approach to calibrating error rates when using a Fellegi-Sunter model \citep{fellegi_theory_1969}. \cite{Marchant2017} proposed an adaptive importance sampling estimator to estimate precision and recall from sampled record pairs. Other approaches to the estimation of performance metrics are model based, where estimated precision and recall can be obtained from predicted match probabilities between record pairs \citep{enamorado2019using}. However, these model-based approaches cannot be used when working with black-box machine learning models or ad hoc clustering algorithms.

{In contrast, our approach to estimation is more practical than pairwise sampling and applies to any black-box disambiguation algorithms such as those used at PatentsView.}

\subsubsection{Our Approach}\label{sec:approach}

Our approach to estimating performance metrics is based on the use of benchmark datasets that already exist or that can be collected in a cost-effective way. These datasets contain entity clusters corresponding to either: (a) sampling records and recovering all associated instances, (b) directly sampling clusters, or (c) sampling blocks. For each of these sampling processes, we propose estimators that correct for the issues discussed in section \ref{sec:challenges}, are nearly unbiased, and are easy to use in practice.

Our approach has the following advantages:
\begin{enumerate}
    \item It can leverage existing benchmark datasets as well as new datasets collected specifically for performance evaluation.
    \item It can easily be generalized to estimate other clustering metrics, such as cluster precision, cluster recall, cluster homogeneity and completeness, and other generalized merge distances.
    \item For evaluation, the review of entity clusters is much more efficient than the review of record pairs. We can achieve high accuracy with small samples without relying on sophisticated sampling schemes.
\end{enumerate}

{Furthermore, our approach is novel. To our knowledge, we are the first to propose unbiased performance estimators based on cluster and block samples. Past work either ignored biases when computing precision and recall from benchmark datasets \citep{frisoli2018exploring, Monath2021, han2019disambiguating}, did not provide estimates for precision or recall \citep{mcveigh2019scaling}, or provided solutions tailored to very specific record linkage models \citep{Belin1995}. We provide the first general solution to entity resolution evaluation that does not rely on sampling record pairs and that applies to any disambiguation algorithm.}

In short, the proposed approach is simple, principled, and practical. It is simple to use, it is statistically principled in its account of sampling processes and uncertainty, and it is practical in the way that it can provide cost-effective estimates for any disambiguation algorithm.

\subsection{Structure of the Paper}

The rest of the paper is organized as follows. In section \ref{sec:data_methods}, {we describe benchmark datasets, our hand-disambiguation methodology for evaluation, the proposed estimators, and our simulation study that we use to validate the performance of our estimators.} Section \ref{sec:results} then presents our performance estimates and results from the simulation study. Section \ref{sec:discussion} summarizes the paper and explores future research directions.

\section{Data and Methodology}\label{sec:data_methods}

In this section, {we introduce the benchmark datasets used at PatentsView, our hand-disambiguation methodology, our proposed performance metric estimators, and the simulation framework that we use to compare estimators.} {Note that we focus on \textit{inventor} disambiguation throughout, rather than on the related problems of assignee and location disambiguation.}

\subsection{Benchmark Datasets for Inventor Disambiguation}

We consider the following benchmark datasets for inventor disambiguation.

\paragraph{Israeli Inventors Benchmark}
\cite{trajtenberg2008identification} disambiguated the U.S. patents of Israeli inventors that were granted between 1963 and 1999. A total of 6,023 Israeli inventors were identified for this time period with 15,310 associated patents.

\paragraph{\cite{Li2014}'s Inventors Benchmark}

Based on an original dataset from \cite{gu2008march}, \cite{Li2014} disambiguated the patent history (between 1975 and 2010) of 95 U.S. inventors.

\subsection{Hand-Disambiguation Methodology}\label{sec:hand-disambiguation}

In addition to considering the above benchmark datasets, {we have carried out hand-disambiguation of inventor mentions. This was motivated by the evaluation of the current PatentsView inventor disambiguation using the estimators proposed in section \ref{sec:proposed_estimators}.}

In total, 100 inventors were sampled with probability proportional to their number of granted patents. This was done by sampling inventor mentions uniformly at random and recovering all patents for a given inventor. These inventor mentions were from U.S. patents granted between 1976 and December 31, 2021.

Two AIR staff were tasked with recovering inventors' patents given sampled inventor mentions. First, given a sampled inventor, the associated predicted cluster was reviewed \ob{and any wrongly} assigned patents were removed. Next, PatentsView's search tools were used to find additional mentions of similarly named inventors. These inventor mentions were reviewed \ob{and added} to the predicted cluster, if appropriate. The two AIR staff had an initial training session, followed by a test run on 10 inventors, before carrying out the rest of the data collection. They worked independently, which resulted in two datasets being obtained for the same inventor mentions. In section \ref{sec:results}, these are referred to as the \textbf{Staff 1} and \textbf{Staff 2} datasets.

Note that our data collection methodology is biased toward PantentsView's current disambiguation. {Indeed, we did not expect the staff to have found all errors or all missing inventor mentions from the predicted clusters. The staff used their best judgment, supported by a thorough search, to resolve inventor mentions.} In cases where no errors were found, the current disambiguation was assumed to be correct. Performance estimates based on this data might therefore be slightly optimistic, which should be acknowledged when reporting performance estimates to PatentsView.org users. Otherwise, for the purpose of improving the current disambiguation \ob{algorithm, this data} is still appropriate to use. It represents the most visible errors in the current disambiguation rather than the totality of them.

\subsection{Proposed Performance Estimators}\label{sec:proposed_estimators}

Throughout the rest of the paper, we focus on pairwise precision and pairwise recall (defined below in \eqref{eq:def_P_R}) as our performance evaluation metrics.

\subsubsection{Representation Lemmas}

First, we define pairwise precision and recall in terms of the number of links between records. Let $\mathcal{D} = \{1, 2, 3, \dots, N\}$ index a set of records let $\mathcal{C}$ be the partition of $\mathcal{D}$ representing ground truth clustering, and let $\widehat{\mathcal{C}}$ be a set of predicted clusters. Now let $\mathcal{T}$ be the set of record pairs that appear in the same cluster in $\mathcal{C}$ (matching pairs), and let $\mathcal{P}$ be the set of record pairs that appear in the same predicted cluster in $\widehat{\mathcal{C}}$ (predicted links). Pairwise precision ($P$) and pairwise recall ($R$) are then defined as
\begin{equation}\label{eq:def_P_R}
    P = \frac{\lvert \mathcal{T} \cap \mathcal{P} \rvert}{\lvert \mathcal{P} \rvert}, \quad R = \frac{\lvert \mathcal{T} \cap \mathcal{P} \rvert}{\lvert \mathcal{T} \rvert}.
\end{equation}

Note that $P = R \, \lvert \mathcal{T} \rvert / \lvert \mathcal{P} \rvert$. As such, precision and recall are equal if and only if the right number of matching pairs is predicted under $\widehat{\mathcal{C}}$.

\begin{figure}[h]
    \centering
    \includegraphics{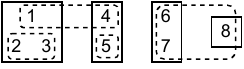}
    \caption{\ob{Example of ground truth clustering $\mathcal{C}$ (represented by boxes with a full border) and predicted clustering $\widehat{\mathcal{C}}$ (rounded boxes with a dotted border) of elements $\mathcal{D} = \{1,2,\dots, 8\}$. Here $\mathcal{T} = \{(1,2), (2,3), (1,3), (4,5), (6,7)\}$ and $\mathcal{P} = \{(1,4), (2,3), (6,7), (7,8), (6,8)\}$. As such, $P = R = 2/5$ in this example.}}
    \label{fig:my_label}
\end{figure}

We now provide three alternative representations of precision and recall that correspond to the processes of sampling records, sampling true clusters, and sampling blocks. These representations will be used to obtain precision and recall estimators under these sampling processes.

\paragraph{Record Sampling Representation}
For a given $i \in \mathcal{D}$, let $c(i) \in \mathcal{C}$ be the ground truth cluster associated with $i$ in $\mathcal{C}$. For a given $c \in \mathcal{C}$, we define
\begin{equation*}
    f(c,\, \widehat{\mathcal{C}}) = \sum_{\hat c \in \hat{\mathcal{C}}} { \lvert c \cap \hat c \rvert \choose 2 }, \quad g(c,\, \widehat{\mathcal{C}}) = \frac{f(c,\, \widehat{\mathcal{C}}) }{\sum_{\hat c \in \hat{\mathcal{C}}} { \lvert \hat c \rvert \choose 2 }}.
\end{equation*}
\begin{lemma}\label{lemma:record_sampling_representation}
If $i$ is distributed over $\mathcal{D}$ with probabilities $p_i > 0$, then
\begin{equation*}
    P = \mathbb{E} \left[ \frac{g(c(i),\, \widehat{\mathcal{C}})}{p_i \lvert c(i) \rvert} \right], \quad R = 2\frac{\mathbb{E}\left[ f(c(i),\, \widehat{\mathcal{C}})  \big/ (\lvert c(i) \rvert p_i) \right]}{\mathbb{E}\left[ (\lvert c(i) \rvert - 1) / p_i \right]}.
\end{equation*}
\end{lemma}
\begin{proof}
By breaking down $\mathcal{P}$ and $\mathcal{T}$ over predicted clusters, we find
\begin{equation*}
    P = \sum_{c\in\mathcal{C}} \sum_{\hat c \in \widehat{\mathcal{C}}} {\lvert c \cap \hat c \rvert \choose 2} \Big/ \sum_{\hat c \in \widehat{\mathcal{C}}} {\lvert \hat c \rvert \choose 2}.
\end{equation*}
Now writing $\sum_{c \in \mathcal{C}} \sum_{\hat c \in \widehat{\mathcal{C}}}{\lvert c \cap \hat c \rvert \choose 2} = \sum_{i=1}^{N} \frac{1}{\lvert c(i) \rvert} \sum_{\hat c \in \widehat{\mathcal{C}}} {\lvert c(i) \cap \hat c \rvert \choose 2}$ and substituting $g(c(i), \widehat{\mathcal{C}})$, we obtain
\begin{equation*}
    P = \sum_{i=1}^N p_i \frac{g(c(i)), \widehat{\mathcal{C}})}{p_i \lvert c(i) \rvert} = \mathbb{E}\left[ \frac{g(c(i), \widehat{\mathcal{C}})}{p_i \lvert c(i) \rvert} \right].
\end{equation*}
For recall, write 
\begin{equation}\label{eq:recall_cluster_form}
    R = \sum_{c\in\mathcal{C}} \sum_{\hat c \in \widehat{\mathcal{C}}} {\lvert c \cap \hat c \rvert \choose 2} \Big/ \sum_{ c \in {\mathcal{C}}} {\lvert c \rvert \choose 2}.
\end{equation}
Through a similar argument as above, we may express the numerator as 
\begin{equation}\label{eq:nominator}
    \sum_{c\in\mathcal{C}} \sum_{\hat c \in \widehat{\mathcal{C}}} {\lvert c \cap \hat c \rvert \choose 2} = \mathbb{E}\left[ f(c(i),\, \widehat{\mathcal{C}})  \big/ (\lvert c(i) \rvert p_i) \right].
\end{equation}
For the denominator, we have
\begin{equation}\label{eq:denominator}
    \sum_{c \in \mathcal{C}} {\lvert c \rvert \choose 2} = \sum_{i=1}^N \frac{1}{\lvert c(i) \rvert}  {\lvert c(i) \rvert \choose 2} = \mathbb{E}[ (\lvert c(i) \rvert - 1)/(2 p_i)].
\end{equation}
Combining \eqref{eq:nominator} and \eqref{eq:denominator} yields the result. 
\end{proof}

\paragraph{Cluster Sampling Representation}

In the cluster sampling case, sampling probabilities are typically known only up to a normalizing factor. This is because the total number of true clusters and other aspects of the ground truth cluster distribution are unknown in practice. As such, we provide expressions for precision and recall that only require knowing the sampling probabilities up to a normalizing factor. This allows the consideration of sampling uniformly at random and sampling clusters with probability proportional to their size. 

\begin{lemma}\label{lemma:cluster_sampling_representation}
If $c$ is distributed over $\mathcal{C}$ with probabilities proportional to $p_c > 0$, then
\begin{equation}\label{eq:lemma_1}
    P = \frac{N\, \mathbb{E} \left[ g(c,\, \widehat{\mathcal{C}})/p_c \right]}{\mathbb{E} \left[ \vert c \rvert /p_c \right]}, \quad R = \frac{\mathbb{E}\left[ f(c, \widehat{\mathcal{C}})/p_c \right]}{\mathbb{E}\left[ {\lvert c \rvert \choose 2}/p_c \right]}.
\end{equation}
\end{lemma}
\begin{proof}
Let $\pi > 0$ be such that $\sum_{c \in \mathcal{C}} \pi p_c= 1$. Now write
\begin{equation*}
    P = \sum_{c \in \mathcal{C}} g(c, \widehat{\mathcal{C}}) = \lvert \mathcal{C} \rvert\, \sum _{c \in \mathcal{C}} \pi p_c g(c, \widehat{\mathcal{C}}) / (\pi p_c \lvert \mathcal{C} \rvert)  =  \lvert \mathcal{C}\rvert\, \mathbb{E}\left[ g(c, \widehat{\mathcal{C}})/(\pi p_c \lvert \mathcal{C} \rvert) \right]
\end{equation*}
and 
\begin{equation*}
    \lvert \mathcal{C} \rvert = \frac{N}{\frac{1}{\lvert \mathcal{C} \rvert} \sum_{c \in \mathcal{C}} \lvert c \rvert} = \frac{N}{\mathbb{E}[\lvert c \rvert / (\pi p_c \lvert \mathcal{C} \rvert)]}.
\end{equation*}
Simplifying $\pi \lvert \mathcal{C} \rvert$ from the numerator and denominator then yields the expression for precision.

The expression for recall follows in a straightforward way from \eqref{eq:recall_cluster_form}.
\end{proof}

Note that, with clusters sampled with probability proportional to their size, the expression for precision simplifies to
$
    P = N\, \mathbb{E} \left[ g(c,\, \widehat{\mathcal{C}})/\lvert c \rvert \right].
$

\begin{remark}
    Lemma \ref{lemma:record_sampling_representation} and lemma \ref{lemma:cluster_sampling_representation} can be generalized to apply to any performance metric that can be expressed as a sum $\sum_{c \in \mathcal{C}} h(c, \widehat{\mathcal{C}})$, for some function $h$, or as a function of such sums. For instance, the use of the so-called \textit{cluster} precision and \textit{cluster} recall \citep{Barnes2015}, or of cluster homogeneity and completeness \citep{Barnes2015}, can be more appropriate in the presence of large clusters. We leave these generalizations as extensions of our work.
\end{remark}

\paragraph{Disjoint Block Sampling Representation}

Let $\mathcal{B}$ be a partition of $\mathcal{D}$ such that for every $ c \in \mathcal{C}$, there exists $b \in \mathcal{B}$ with $c \subset b$. For a given $b \in \mathcal{B}$, let $\mathcal{T}_b$ be the set of ground truth links contained within $b$, let $\mathcal{P}_b$ be the set of predicted links contained within $b$, and let $\mathcal{P}_b^{-}$ be the set of predicted links with a single record in $b$ (i.e., $\mathcal{P}_b^{-}$ is the set of outgoing links from $b$).

\begin{lemma}\label{lemma:block_sampling_representation}
If $b$ is distributed over $\mathcal{B}$ with probabilities proportional to $p_b > 0$, then
\begin{equation}\label{eq:lemma_3}
    P = \frac{\mathbb{E}\left[ \lvert \mathcal{T}_b \cap \mathcal{P}_b \rvert / p_b \right]}{\mathbb{E}\left[ (\lvert \mathcal{P}_b \rvert + \tfrac{1}{2} \lvert \mathcal{P}_b^{-} \rvert)/p_b \right]}, \quad R = \frac{\mathbb{E}\left[ \lvert \mathcal{T}_b \cap \mathcal{P}_b \rvert / p_b \right]}{\mathbb{E}\left[ \lvert \mathcal{T}_b \rvert / p_b \right]}.
\end{equation}
\end{lemma}

\begin{proof}
Since the blocking procedure is assumed to have no error (for every $ c \in \mathcal{C}$, there exists $b \in \mathcal{B}$ with $c \subset b$), we can break down $\mathcal{T}$ as the disjoint union of the $\mathcal{T}_b$'s over $b \in \mathcal{B}$. It follows that $\lvert T \rvert = \sum_{b \in \mathcal{B}} \lvert \mathcal{T}_b \rvert$ and  $\lvert \mathcal{T} \cap \mathcal{P} \rvert = \sum_{b \in \mathcal{B}} \lvert \mathcal{T}_b \cap \mathcal{P}_b \rvert$. The expression for recall in \eqref{eq:lemma_3} follows directly. For precision, we can express $\lvert \mathcal{P} \rvert$ as the number of links within blocks plus the number of links across blocks. Since the number of links across blocks is counted twice when each block is considered, we obtain $\lvert \mathcal{P} \rvert = \sum_{b \in \mathcal{B}} \left( \lvert \mathcal{P}_b \rvert + \tfrac{1}{2} \lvert \mathcal{P}_b^{-} \rvert \right)$. 
\end{proof}

\begin{remark}
    Lemmas \ref{lemma:record_sampling_representation} -- \ref{lemma:block_sampling_representation} are formulated in terms of sampled ground truth clusters and sampled blocks that do not contain errors. However, given the duality between precision and recall (interchanging the roles between $\mathcal{C}$ and $\widehat{\mathcal{C}}$  interchanges precision and recall), the results also apply to sampling \textit{predicted} clusters.
\end{remark}

\subsubsection{Proposed Estimators}

All of the expressions for precision and recall in lemmas \ref{lemma:record_sampling_representation} -- \ref{lemma:block_sampling_representation} are either population means or ratios of population means. As such, they can be estimated using sample means and ratios of sample means. For readability, we present here a generic formula for an approximately unbiased estimator of the ratio of means and then specify the needed quantities for each representation below. The estimator applies a first order bias correction to the ratio of sample means, based on a Taylor approximation. Approximate confidence intervals can be computed based on the variance estimator of the ratio of sample means by Taylor approximation, assuming the corrected estimator has a small bias \citep{sarndal2003model, fuller2011sampling}. The generic formula for estimating the ratio of $T$-sized ``population'' (of records/clusters/blocks) means of the form
\begin{align}
    E=\frac{1/T\sum_{i=1}^{T} B_i}{1/T\sum_{i=1}^T A_i},
\end{align}
assuming we have sampled $n$ elements (records/clusters/blocks), is
\begin{align}\label{eq:bias_adjustment}
    \widehat{E} = \frac{\bar{B}_n}{\bar{A}_n} \left\{ 1 + \frac{\theta_{n,T}}{n(n-1)} \sum_{s=1}^n \frac{A_s}{\bar{A}_n} \left(\frac{B_s}{\bar{B}_n} - \frac{A_s}{\bar{A}_n} \right) \right\}, \; \bar{A}_n = \frac{1}{n} \sum_{s=1}^{n} A_s, \; \bar{B}_n = \frac{1}{n} \sum_{s=1}^{n} B_s
\end{align}
We note that an additional symbol $\theta_{n,T}$ is introduced for a possible finite population correction when relatively large number of elements are sampled without replacement (see below). Classical adjustment is set to $\theta_{n,T} = (1 - \frac{n-1}{T-1})$ \citep{Cochran1977}. For practical purposes when $T$ is large, $\theta_{n,T}=1$ will suffice. In fact, knowledge of $T$ is not needed at all as long as it is large enough relative to $n$, which is useful because the total number of true clusters/blocks is not known in advance. Confidence intervals can be computed based on the variance estimate of the above, which is
\begin{align} \label{eq:estimated_variance}
    \widehat{V}(\widehat{E}) = \left( \frac{\bar{B}_n}{\bar{A}_n} \right)^2 \frac{\theta_{n,T}}{n(n-1)} \sum_{s=1}^n \left( \frac{A_s}{\bar{A}_n} - \frac{B_s}{\bar{B}_n}\right)^2.
\end{align}
The specific values for $A_s$, $B_s$, and $T$ in each of our representation are described in appendix \ref{sec:appendix_estimators}.


\begin{remark}
    In entity resolution applications, we can typically assume that elements have been sampled with replacement (or closely so). Indeed, with a small proportion of sampled elements, nonreplacement samples are approximately equivalent to samples with replacement. However, if dealing with relatively large nonreplacement samples, then the sampling probabilities used in the definition of the estimators should be adjusted to reflect the size-dependent effect of nonreplacement \citep{Horvitz1952}.
\end{remark}

\subsection{Simulation Study}\label{sec:simulation_study}

In order to assess the performance of the proposed estimators, we carried out a simulation study based on PatentsView's inventor disambiguation. Specifically, in the context of the simulation, we considered PatentsView current inventor disambiguation as the ground truth clustering. A simulated set of predicted clusters was obtained by introducing errors (misattribution of inventor mentions) into the current disambiguation. We then estimated the precision and recall of this predicted clustering using our estimators based on random cluster samples. The process of sampling clusters and estimating precision/recall was repeated $100$ times in order to provide the distribution of the estimators and metrics such as bias and root mean squared error (rmse). 

To introduce errors, we picked records at random and changed their cluster assignment to that of other records picked at random. This is a simple process that ensures that larger clusters are more likely to contain errors. In our simulation, we considered rates of $5\%$, $10\%$, and $25\%$ for the proportion of records that are sampled for cluster misassignment. Although the larger error rates are more realistic, the $5\%$ misattribution rate helps showcase the properties of our estimators when only a small proportion of the sampled clusters is associated with errors.

For the sampling process, we considered sampling records uniformly at random and recovering their associated clusters. This is the same as sampling clusters with probability proportional to their size. In the record/cluster sampling cases, we looked at the effect of sampling $100$, $200$ and $400$ records/clusters.

Finally, we compared the following three precision and recall estimators:
\begin{description}
    \item[\texttt{P\_naive}, \texttt{R\_naive}] This is the ``naive" precision (respectively recall) estimator obtained by computing precision (respectively recall) when only looking at records that appear in the sampled clusters.
    \item[\texttt{P\_record}, \texttt{R\_record}] These are the precision and recall estimators corresponding to uniformly sampling records in lemma \ref{lemma:record_sampling_representation} ($p_i \propto 1$) with the bias adjustment given in \eqref{eq:bias_adjustment}. Note that these are the same as the estimators obtained from lemma \ref{lemma:cluster_sampling_representation} when sampling clusters with probability proportional to their size ($p_c \propto \lvert c \rvert$).
    \item[\texttt{P\_cluster\_block}] This is the precision estimator obtained by considering each sampled cluster as its own block in lemma \ref{lemma:block_sampling_representation}, where clusters have been sampled with probability proportional to their size ($p_b \propto \lvert b \rvert$) and with the bias adjustment given in \eqref{eq:bias_adjustment}. Note that in the case of recall with cluster blocks, the estimator corresponding to lemma \ref{lemma:block_sampling_representation} is the same as the one corresponding to lemma \ref{lemma:cluster_sampling_representation}.
\end{description}

\section{Results}\label{sec:results}

\subsection{Results From the Simulation Study}

Figure \ref{fig:sim_records_precision} shows the distribution of the three precision estimators used in the simulation study (see section \ref{sec:simulation_study}) compared to ground truth precision. The block sampling estimator \texttt{P\_cluster\_block} is highly accurate, while \texttt{P\_record} is more variable and \texttt{P\_naive} is entirely uninformative. Note that \texttt{P\_record} can take values greater than $1$ (not shown in this figure), so that truncating it to be less than $1$ introduces a bias in some cases.

\texttt{P\_cluster\_block} performs better than \texttt{P\_record} because \texttt{P\_record} has been derived in a generic way that applies to any performance metric that can be expressed in cluster form similar to \eqref{eq:lemma_1}. On the other hand, \texttt{P\_cluster\_block} relies on specific properties of precision. Among other things, this ensures that \texttt{P\_cluster\_block} is constrained to be between $0$ and $1$. In practice, \texttt{P\_cluster\_block} should be preferred as a pairwise precision estimator. The bias and rmse of the precision estimators are reported in table \ref{tab:precision_results}.

\begin{figure}
    \centering
        \caption{Distribution of precision estimates for various sample sizes and misattribution rates, as described in section \ref{sec:simulation_study}. Ground truth precision is marked by a dotted vertical line. The \texttt{rate} variable represents the percentage misattribution rate. The estimator \texttt{P\_cluster\_block} is highly accurate, while \texttt{P\_naive} is almost always close to $1.0$, having little to do with the true precision.}
    \includegraphics[width=\linewidth]{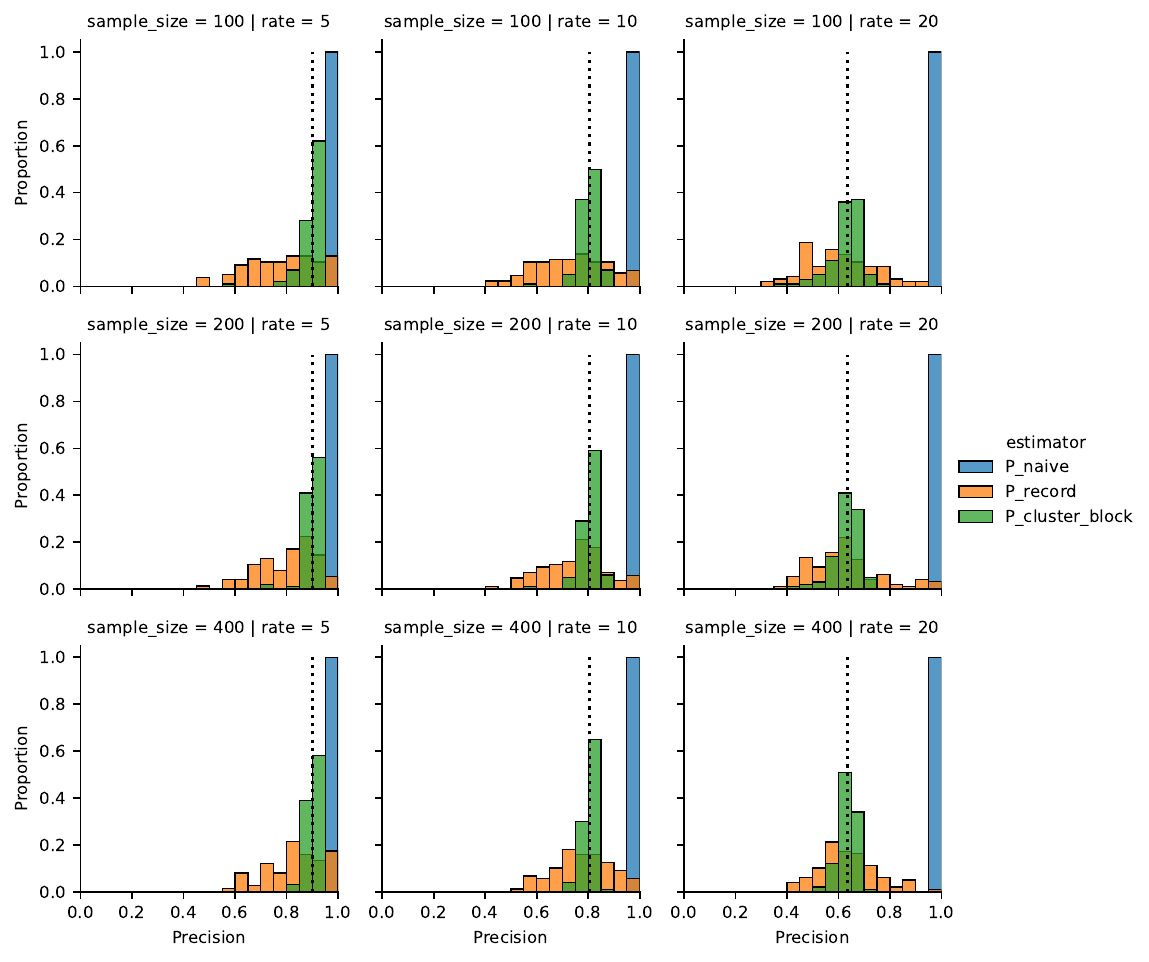}
    \label{fig:sim_records_precision}
\end{figure}

\begin{table}
    \centering
    \caption{\label{tab:precision_results}\centering Bias and root mean squared error (rmse) of precision estimators for the simulation study described in section \ref{sec:simulation_study}. The \texttt{rate} variable represents the percentage misattribution rate.}
\begin{tabular}{llccccccccc}
\toprule
{} & \hfill{\textbf{rate}} & \multicolumn{3}{c}{5} & \multicolumn{3}{c}{10} & \multicolumn{3}{c}{20} \\
\cmidrule(lr){3-5}\cmidrule(lr){6-8}\cmidrule(lr){9-11}
{} & \hfill{\textbf{sample size}} & {100} & {200} & {400} & {100} & {200} & {400} & {100} & {200} & {400} \\
\midrule
\multirow[c]{4}{*}{\textbf{bias}} & P\_cluster\_block & -0.004 & \textbf{-0.002} & \textbf{-0.001} & -0.002 & \textbf{-0.001} & \textbf{0.001} & -0.004 & \textbf{-0.001} & \textbf{-0.000} \\
 & P\_naive & 0.099 & 0.099 & 0.099 & 0.195 & 0.195 & 0.195 & 0.366 & 0.366 & 0.365 \\
 & P\_record & \textbf{-0.002} & 0.013 & 0.009 & \textbf{-0.001} & 0.011 & 0.008 & \textbf{-0.001} & 0.009 & 0.006 \\
\midrule
\multirow[c]{4}{*}{\textbf{rmse}} & P\_cluster\_block & \textbf{0.045} & \textbf{0.033} & \textbf{0.021} & \textbf{0.041} & \textbf{0.038} & \textbf{0.026} & \textbf{0.063} & \textbf{0.052} & \textbf{0.034} \\
 & P\_naive & 0.099 & 0.099 & 0.099 & 0.195 & 0.195 & 0.195 & 0.366 & 0.366 & 0.365 \\
 & P\_record & 0.294 & 0.232 & 0.169 & 0.260 & 0.205 & 0.150 & 0.207 & 0.162 & 0.117 \\
\bottomrule
\end{tabular}
\end{table}

Regarding recall, figure \ref{fig:recall_results} shows the distribution of the two estimators used in the simulation study. The naive recall estimator performs well in this case. However, the recall estimator accounting for the sampling mechanism is more accurate (\texttt{R\_record}, which is equal to \texttt{R\_cluster\_block}). The bias and rmse of the recall estimators are reported in table \ref{tab:recall_results}, {where the overall improved performance of \texttt{R\_record} can be observed.}

\begin{figure}
    \centering
        \caption{Distribution of recall estimates for various sample sizes and misattribution rates, as described in section \ref{sec:simulation_study}. Ground truth recall is marked by a dotted vertical line. The \texttt{rate} variable represents the percentage misattribution rate.}
    \includegraphics[width=\linewidth]{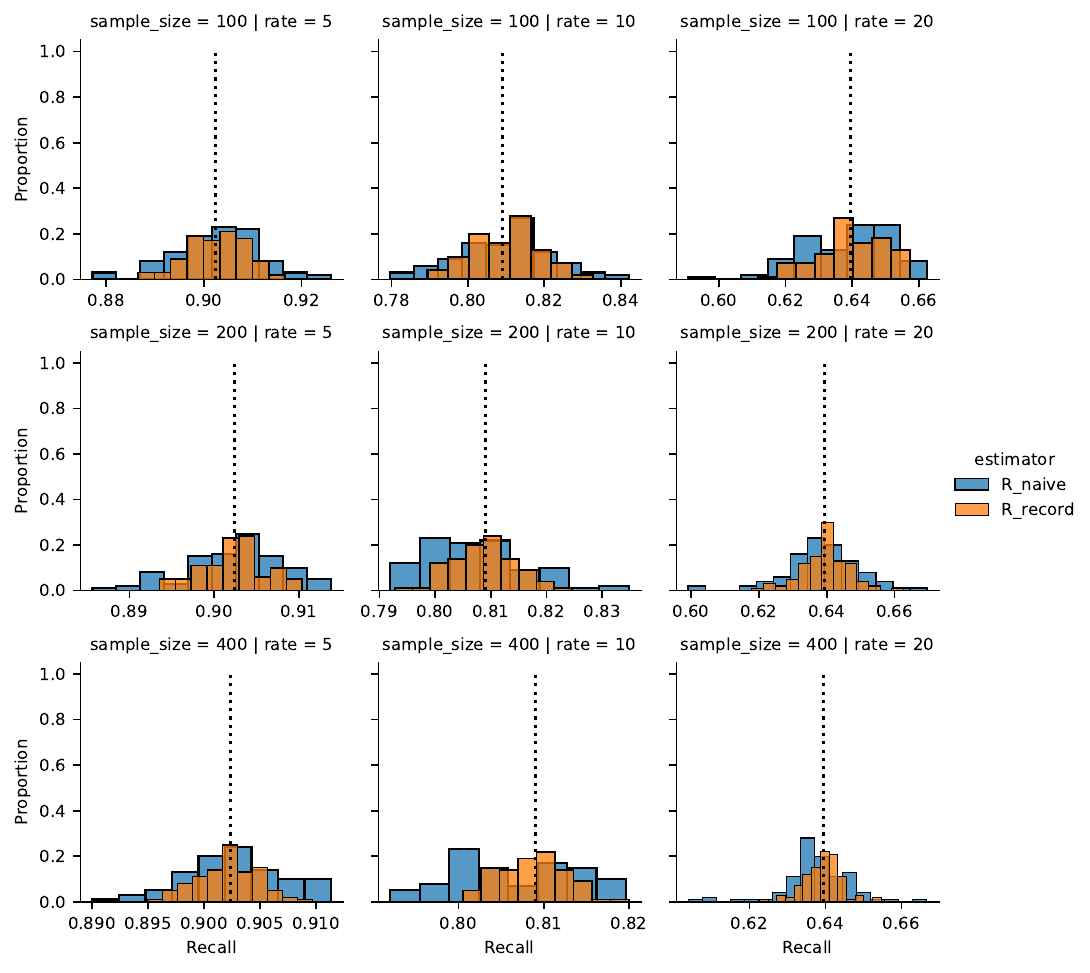}
    \label{fig:recall_results}
\end{figure}

\begin{table}
    \centering
    \caption{\label{tab:recall_results}\centering Bias and root mean squared error of recall estimators for the simulation study described in section \ref{sec:simulation_study}. The \texttt{rate} variable represents the percentage misattribution rate.}
\begin{tabular}{llccccccccc}
\toprule
{} & {\textbf{rate}} & \multicolumn{3}{c}{{5}} & \multicolumn{3}{c}{10} & \multicolumn{3}{c}{20} \\
\cmidrule(lr){3-5}\cmidrule(lr){6-8}\cmidrule(lr){9-11}
{} & {\textbf{sample size}} & {100} & {200} & {400} & {100} & {200} & {400} & {100} & {200} & {400} \\
\midrule
\multirow[c]{2}{*}{\textbf{bias}} & R\_record & \textbf{0.0002} & -0.0001 & \textbf{-0.0002} & 0.0011 & \textbf{-0.0002} & \textbf{-0.0004} & \textbf{0.0005} & \textbf{0.0001} & \textbf{-0.0002} \\
 & R\_naive & 0.0007 & 0.0001 & 0.0004 & \textbf{0.0009} & -0.0015 & -0.0024 & -0.0008 & -0.0010 & -0.0017 \\
 \midrule
 \multirow[c]{2}{*}{\textbf{rmse}} & R\_record & \textbf{0.0060} & \textbf{0.0037} & \textbf{0.0028} & \textbf{0.0089} & \textbf{0.0056} & \textbf{0.0039} & \textbf{0.0099} & \textbf{0.0071} & \textbf{0.0050} \\
 & R\_naive & 0.0089 & 0.0056 & 0.0043 & 0.0123 & 0.0089 & 0.0075 & 0.0132 & 0.0109 & 0.0091 \\

\bottomrule
\end{tabular}

\end{table}

\subsection{Evaluation of PatentsView's Disambiguation}

Table \ref{tab:results} shows estimated pairwise precision and recall from benchmark datasets and from our two hand-curated datasets. Note that each estimate is associated with a given population of inventor mentions which is a subset of granted U.S. patents since 1976. We focused on U.S. patents granted since 1976 as this is the main data product of PatentsView.org.

The choice of estimators for the results presented in table \ref{tab:results} is as follows.
Since our two hand-curated datasets (see section \ref{sec:hand-disambiguation}) were obtained by sampling inventor clusters with probabilities proportional to cluster sizes, we used the \texttt{P\_cluster\_block} and \texttt{R\_record} estimators with corresponding probability weights. 
For the Israeli benchmark dataset, we assumed that a single block of inventor clusters was sampled and used the corresponding estimators defined through \eqref{eq:lemma_3} with a single block sample. Note that no variance estimates can be given for single samples.
For \cite{Li2014}'s inventors benchmark, given that the inventor clusters were originally obtained from a set of inventor curriculum vitae, we assumed that the clusters were sampled uniformly at random. As such, we used cluster block estimators with constant probability weights.

\begin{table}[!h]
    \centering
    \caption{\centering Estimated pairwise precision and recall (with estimated standard deviation) from our four benchmark datasets.}
    \label{tab:results}
    \begin{tabular}{lccc}
    \toprule
    \textbf{dataset} & \textbf{est. precision ($\hat \sigma$)} & \textbf{est. recall ($\hat \sigma$)} & \textbf{scope} \\
    \midrule
    Staff 1 & 88\% (3.4\%) & 95\% (1.1\%)  & 1976 -- Dec. 31, 2022 (U.S. granted)\\
    Staff 2 & 87\% (3.6\%) & 96\% (1.0\%) & 1976 -- Dec. 31, 2022 (U.S. granted)\\
    Israeli Benchmark & 79\% (NA) & 94\% (NA) & 1976 -- 1999 (U.S. granted)\\
    \cite{Li2014}'s Benchmark & 91\% (2.7\%) & 91\% (5.0\%) & 1976 -- 2010 (U.S. granted)\\
    \bottomrule
    \end{tabular}
\end{table}

Overall, our performance estimates paint the first realistic picture of PatentsView's disambiguation accuracy in practice. Precision is not nearly 100\%, as would be assumed from naively computing precision on benchmark datasets. Rather, there is significant room for improvement. Our hand-curated datasets and data collection methodology provide the necessary basis to investigate errors and plan for improvements to the disambiguation algorithm.

\section{Discussion}\label{sec:discussion}

Motivated by PatentsView's disambiguation, this paper introduced a novel evaluation methodology for entity resolution algorithms. The methodology relies on benchmark datasets containing ground truth clusters and estimators that account for biases inherent to these datasets. For PatentsView, this provided the first representative estimates of its disambiguation performance. Furthermore, all data and code used in this paper, as well as other tools developed to facilitate evaluation at PatentsView, are freely available at \hyperref[https://github.com/PatentsView/PatentsView-Evaluation/releases/tag/1.0.1]{https://github.com/PatentsView/PatentsView-Evaluation/releases/tag/1.0.1}. An updated software package is described in \cite{Binette2023}.

There are two main products resulting from this work. The first is the appropriate understanding of the quality of the data provided to PatentsView's users. Our performance estimates indicate that, despite an overall accurate disambiguation, there is significant room for \ob{improvement. Notably, the current disambiguation over-estimates the number of matching inventor mention pairs.} The second product is the set of tools needed for methodological research and model comparison. Given our evaluation methodology, we can now reliably compare algorithms and decide with confidence on changes that will affect users.

{One important topic for future work is the quantification of uncertainty associated with errors in the hand-disambiguation process. This is challenging problem given the lack of validation information available. Surveying inventors to validate the hand-disambiguation process would be one way to explore this issue. Sensitivity analyses involving potential errors in the hand-disambiguation process could also be informative. We refer the reader to \cite{bailey2017well} for a state-of-the-art hand-labeling study that evaluated human review accuracy.}

{Another important topic for future work is the development of estimators for additional performance metrics. As we have seen in the simulation study, the generically derived estimators (e.g., from lemma \ref{lemma:cluster_sampling_representation}) do not perform as well as estimators derived using specific properties of pairwise precision and recall (lemma \ref{lemma:block_sampling_representation}). As such, care should be taken to obtain efficient estimators for every metric of interest.}

\ob{Finally, we note that the performance of estimators can degrade when dealing with heavy-tailed cluster size distributions. The bias of ratio estimators can be high in this case, especially when using small samples and when sampling clusters uniformly at random rather than with probability proportional to size. Model-based estimators that exploit known properties of the cluster size distribution could be developed to improve estimation accuracy in such cases.}

\section*{Data and Code}

All data and code used for this paper are available as part of the PatentsView-Evaluation Python package (version 1.0.1) at \hyperref[https://github.com/PatentsView/PatentsView-Evaluation/releases/tag/1.0.1]{https://github.com/PatentsView/PatentsView-Evaluation/releases/tag/1.0.1}.


\section*{Author Contributions}

Olivier Binette led the evaluation project and wrote the majority of the manuscript. Sokhna A York and Emma Hickerson carried out the data collection by manually reviewing inventor clusters. Youngsoo Baek provided bias adjustment and uncertainty quantification for ratio estimators. Sarvo Madhavan was a technical advisor and contributed to code. Christina Jones was an advisor and project manager. All authors provided input on the manuscript.

\section*{Disclosure Statement}

The authors report there are no competing interests to declare.

\nocite{Subramanian2021}
\bibliographystyle{chicago}
\bibliography{biblio}

\appendix

\section{Appendix}\label{sec:appendix}

\subsection{Bias of Precision Computed on Benchmark Datasets}\label{appendix:example}

This section provides more information on example \ref{first_example}.

For this example, we considered the RLdata10000 dataset from \cite{RecordLinkage}. This is a synthetic dataset containing 10,000 records with first name, last name, and date of birth attributes. There is noise in these attributes and a $10\%$ duplication rate. Ground truth identity is known for all records.

The disambiguation algorithm we consider matches records if any of the following conditions are met:
\begin{itemize}
    \item records agree on first name, last name, and birth year,
    \item records agree on first name, birth day, and birth year, or
    \item records agree on last name, birth day, and birth year.
\end{itemize}
Note that this is not at all a good disambiguation algorithm. It has $52\%$ precision and $83\%$ recall. However, it allows us to to showcase the issue with nonadjusted precision computed on cluster samples.

In our experiment, we have repeated 5,000 times the following three-steps process 5,000 times:
\begin{enumerate}
    \item First, 200 records were sampled and the ground truth clusters associated with them were recovered. This step provided a ``benchmark" dataset that was used for evaluation.
    \item Second, a trivial precision estimate was obtained by computing precision over the benchmark dataset. That is, predicted cluster assignments were restricted to records that appear in the benchmark data and precision was compared for these records. More often than not, the result was an observation of $100\%$ precision.
    \item Third, we computed our proposed precision estimator which corresponds to lemma \ref{lemma:record_sampling_representation} and the estimator $\widehat P_{\text{block}}$ defined in \eqref{eq:P_block}, with blocks corresponding to clusters and with sampling probabilities $p_b \propto \lvert b \rvert$.
\end{enumerate}
The distributions of the two precision estimates over the 5,000 repetitions are shown in figure \ref{fig:precision_problem}. Our proposed estimator is accurate and nearly unbiased, whereas the trivial precision estimates have almost nothing to do with actual algorithmic performance.

\subsection{Precision and Recall Estimator Formulas} \label{sec:appendix_estimators}

This section describes specific values for $A_s$, $B_s$, and $T$ in \eqref{eq:bias_adjustment} in order to obtained nearly unbiased precision and recall estimators based on each of the representation in section \ref{sec:proposed_estimators}. We use the symbols $\widehat P$ and $\widehat R$, indexed by either ``$\text{rec}$", ``$\text{clust}$", or ``$\text{block}$", in order to refer to precision and recall estimators corresponding to record, cluster, and block sampling representations, respectively.

\paragraph{Record Sampling Estimators}

\begin{description}
    \item[$\widehat{P}_{\text{rec}}$]{Estimating $P$ is a special case that does not require ratio-of-means estimation. We propose to use a simple unbiased estimator:}
    \begin{align}\label{eq:P_rec}
        \widehat{P}_{\text{rec}} =
        \frac{1}{n}\sum_{s=1}^{n}\frac{g(c(i_s),\widehat{\mathcal{C}})}{|c(i_s)|p_{i_s}}.
    \end{align}
    The unbiased estimator of the variance of $\widehat{P}_{\text{rec}}$ is also available:
    \begin{align}
        \widehat{V}(\widehat{P}_{\text{rec}}) = 
        \frac{\theta_{n,N}}{n(n-1)}\sum_{s=1}^{n}\left(
        \frac{g(c(i_s),\widehat{\mathcal{C}})^2}{|c(i_s)|^2p_{i_s}^2} - 
        \widehat{P}_{\text{rec}}^2.
        \right)
    \end{align}
    \item[$\widehat{R}_{\text{rec}}$]{Going forward, we refer to formulae \eqref{eq:bias_adjustment} and \eqref{eq:estimated_variance}. $\widehat{R}_{\text{rec}}$ and its variance estimate $\widehat{V}(\widehat{R}_{\text{rec}})$ are given by \eqref{eq:bias_adjustment} and \eqref{eq:estimated_variance}, where we substitute in}
    \begin{align}
    T = N,\; A_s = \frac{(\lvert c(i_s) \rvert - 1)}{p_{i_s}},\; 
    B_s = 2\frac{f(c(i_s),\widehat{\mathcal{C}})}{|c(i_s)|p_{i_s}}.
    \end{align}
\end{description}

\paragraph{Cluster Sampling Estimators}

\begin{description}
    \item[$\widehat{P}_{\text{clust}}$]{The estimator and its variance estimate $\widehat{V}(\widehat{P}_{\text{clust}})$ are given by \eqref{eq:bias_adjustment} and \eqref{eq:estimated_variance}, where we substitute in}
    \begin{align}
        T = |\mathcal{C}|,\; A_s = \frac{|c_s|}{p_{c_s}},\; B_s = N\frac{g(c_s,\widehat{\mathcal{C}})}{p_{c_s}}.
    \end{align}
    \item[$\widehat{R}_{\text{clust}}$]{The estimator and its variance estimate $\widehat{V}(\widehat{R}_{\text{clust}})$ are given by \eqref{eq:bias_adjustment} and \eqref{eq:estimated_variance}, where we substitute in}
    \begin{align}
    T = |\mathcal{C}|,\; A_s = \frac{{\lvert c_s \rvert \choose 2}}{p_{c_s}},\; B_s = \frac{f(c_s,\widehat{\mathcal{C}})}{p_{c_s}}. 
    \end{align}
\end{description}

\paragraph{Disjoint Block Sampling Estimators}

\begin{description}
    \item[$\widehat{P}_{\text{block}}$]{The estimator and its variance estimate $\widehat{V}(\widehat{P}_{\text{block}})$ are given by \eqref{eq:bias_adjustment} and \eqref{eq:estimated_variance}, where we substitute in}
    \begin{align}\label{eq:P_block}
        T = |\mathcal{B}|,\; A_s = \frac{\lvert \mathcal{P}_{b_s} \rvert + \frac{1}{2} \lvert \mathcal{P}_{b_s}^{-} \rvert}{p_{b_s}},\;
        B_s = \frac{\lvert \mathcal{T}_{b_s} \cap \mathcal{P}_{b_s} \rvert}{p_{b_s}}.
    \end{align}
    \item[$\widehat{R}_{\text{block}}$]{The estimator and its variance estimate $\widehat{V}(\widehat{R}_{\text{block}})$ are given by \eqref{eq:bias_adjustment} and \eqref{eq:estimated_variance}, where we substitute in}
    \begin{align}
        T = |\mathcal{B}|,\; A_s = \frac{\lvert \mathcal{T}_{b_s} \rvert}{p_{b_s}},\; B_s = 
        \frac{\lvert \mathcal{T}_{b_s} \cap \mathcal{P}_{b_s} \rvert}{p_{b_s}}.
    \end{align}
\end{description}

\end{document}